\documentclass{dmaa}

\begin{document}

\markboth{Authors' Names}
{ILP formulations for triple and quadruple Roman domination problem}

%
\catchline{}{}{}{}{}
%

\title{Integer Linear Programming Formulations for Triple and Quadruple Roman Domination Problems.
}

\author{Sanath Kumar Vengaldas   }

\address{Computer Science and Engineering, National Institute of Technology, Warangal\\
 Telangana, India\\
\email{vengal\_851907@student.nitw.ac.in }}
\author{Adarsh Reddy Muthyala    }

\address{Computer Science and Engineering, National Institute of Technology, Warangal\\
 Telangana, India\\
\email{muthya\_911905@student.nitw.ac.in }}

\author{Bharath Chaitanya Konkati}

\address{Computer Science and Engineering, National Institute of Technology, Warangal\\
 Telangana, India\\
\email{konkat\_851950@student.nitw.ac.in }}
\author{P. Venkata Subba Reddy \\ } 

\address{Computer Science and Engineering, National Institute of Technology, Warangal\\
 Telangana, India\\
\email{pvsr@nitw.ac.in }}
\maketitle
\begin{history}
\received{13 March 2023}
\end{history}

\begin{abstract}

\noindent Roman domination is a well researched topic in graph theory. Recently two new variants of Roman domination, namely triple Roman domination and quadruple Roman domination problems have been introduced, to provide better defense strategies. However, triple Roman domination and quadruple Roman domination problems are NP-hard. In this paper, we have provided genetic algorithm for solving triple and quadruple Roman domination problems. Programming (ILP) formulations for triple Roman domination and quadruple Roman domination problems have been proposed. The proposed models are implemented using IBM CPLEX 22.1 optimization solvers and obtained results for random graphs generated using NetworkX Erdős-Rényi model.
\end{abstract}


\section{Introduction}\label{sec1}

Graphs $G(V,E)$ considered are undirected, connected and simple. Let $V(G)$ represent the vertex set and $E(G)$ represent the edge set of $G$. 
$N(v)$ represents the \textit{open neighbourhood} of vertex $v$, \begin{math} 
N(v)=\{u\mid(u,v) \in E\} 
\end{math} and $N(v) \cup \{v\}$ represents the \textit{closed neighbourhood} $N[v]$ of vertex $v$.

 A  function \begin{math} f  : V\rightarrow \{0,1\} \end{math} on a graph $G$  satisfying the condition that for every vertex $v \in V$ if $f(v)=0$ then $v$ has a neighbour $u$ such that $f(u)=1$ is called a \textit{dominating function} of $G$. \textit{Domination number} (\begin{math} \gamma(G) \end{math}) is the minimum possible sum of $f$ values assigned to all vertices of $G$. Optimization version of domination problem is NP-hard and decision version is NP-complete  \cite{Fundamentals}.

Based on the defense tactics to defend the Roman empire, the concept of Roman domination was introduced in 2004 \cite{roman1},\cite{roman3},\cite{roman2}. A legion from a powerful neighbour might repel an unexpected invasion on any neighbouring undefended city and if such moving would leave its existing city defenceless, no legion could relocate. A function \begin{math} f  : V\rightarrow \{0,1,2\} \end{math} on $G$ is called a \textit{Roman dominating function} if  for every vertex \begin{math} v \in V \end{math}, either $f(v) \in \{1,2\}$ or if $f(v)=0$ then it should have at least 1 neighbour $u$ such that $f(u)=2$. \textit{Roman domination number} (\begin{math} \gamma_{R}(G) \end{math}) is the minimum possible sum of $f$ values assigned to all vertices of $G$. Optimization version of Roman domination problem is NP-hard \cite{NPCRD}. 

Double Roman domination was introduced as a stronger version of weak Roman domination and Roman domination in \cite{double}. A function \begin{math} f  : V\rightarrow \{0,1,2,3\} \end{math} on a graph $G$ is called a \textit{double Roman dominating function} if following conditions are satisfied, for every vertex \begin{math} v \in V \end{math}, if $f(v)=0$, then vertex $v$ must have at least one neighbour $u$ with $f(u)=3$ or at least 2 neighbours assigned value 2 under $f$ , and if $f(v)=1 $, then $v$must have at least 1 neighbour $w$ with \begin{math}f(w) \ge 2\end{math}. A \textit{Double Roman domination number} \begin{math} \gamma_{dR}(G) \end{math} is the minimum possible sum of $f$ values of all vertices of $G$.

In \cite{triple}, \textit{triple Roman domination} problem is introduced because double Roman domination problem is not efficient in some cases as discussed in \cite{double}.A function \begin{math} f  : V\rightarrow \{0,1,2,3,4\} \end{math} on a graph $G=(V,E)$ is called a \textit{triple Roman dominating function} $(3RDF)$ if the following conditions are satisfied,for every vertex \begin{math} v \in V \end{math}, if $f(v)=0$ then $v$ must have at least 3 neighbours assigned value 2 under $f$ or two neighbours $u$,$w$ with  \begin{math} f(u) \ge 2 \end{math} and $f(w)=3$ or one neighbour $z$ with \begin{math} f(z) = 4 \end{math}, and if $f(v)=1 $ then vertex $v$ must have at least 2 neighbours assigned value 2 under $f$ or at least 1 neighbour $w$ such that \begin{math}f(w) \ge 3\end{math} , and if $f(v) = 2$, then vertex $v$ must have at least 1 neighbour $w$  with \begin{math} f(w) \ge 2 \end{math}. Weight of a 3RDF is the sum of function $f$ values of all the vertices. \textit{Triple Roman domination number} of a graph $G$ (\begin{math} \gamma_{3R}(G) \end{math}) is minimum weight of all the possible triple Roman dominating functions. Determining $\gamma_{3R}(G)$ for a graph $G$ is known as triple Roman domination problem (M3RDP).  M3RDP for bipartite graphs and chordal graphs, triple Roman domination problem is NP-hard \cite{triple}.

A function \begin{math} f  : V\rightarrow \{0,1,2,3,4,5\} \end{math} on a graph $G=(V,E)$ is called \textit{quadruple Roman dominating function} (4RDF) if the following conditions are satisfied, for every vertex \begin{math} v \in V \end{math}, if $f(v)=0$ then \begin{math}v\end{math} must have at least 1 neighbour vertex \begin{math}u\end{math}  with $f(u) = 5$ or two vertices $u,w$ such that $f(u) = 4$ and $f(w) \ge 2$ or at least 2 vertices such that $f(u) \ge 3 $ or at least 4 vertices such that $f(u) = 2$ or two vertices with $f(u) = 2$ and one vertex with $f(u) = 3$, and if $f(v) = 1$ then vertex $v$ must have at least 1 neighbour $u$ with $f(u) \ge 4$ or at least 2 vertices $u,w$ such that $f(u) = 3$ and $f(w) \ge 2$ or at least 3 vertices with $f(u) = 2$, and if $f(v) = 2$ then vertex $v$ must have at least 1 neighbour $u$ with $f(u) \ge 3$ or at least 2 vertices with $f(u) = 2$, and if $f(v) = 3$ then it must have at least 1 neighbour $u$ with $f(u) \ge 2$. Weight of a 4RDF is the sum of function $f$ values of all the vertices. \textit{Quadruple Roman domination number} of a graph $G$ (\begin{math} \gamma_{4R}(G) \end{math}) is the minimum weight of all the possible 4RDFs. Determining $\gamma_{4R}(G)$ for a graph $G$ is known as quadruple Roman domination problem (M4RDP). Throughout the paper, by $\gamma_{3R}$-function we mean any 3RDF of graph $G$ with weight $\gamma_{3R}(G)$. Similarly $\gamma_{4R}$-function is defined.

In general $[k]$-Roman dominating function ($k$RDF) of a graph $G$ is also discussed in \cite{triple} and it is a mapping \begin{math} f  : V\rightarrow \{0,1,2,3,\ldots, k+1\} \end{math} such that for every vertex \begin{math} u\in V \ for\ which\ f(u)\le k \end{math} following condition should satisfy.
	    \begin{center}
	        \begin{math}
	            {h(AN[u]) \ge\;\lvert AN(u)\rvert +k}
	        \end{math}
	    \end{center}
	    Here, 
	     \begin{math}
	            AN(u)=\{z:z \in N(u) \; \text{and} \; f(v) \ne 0 \},
	       \end{math}
	       \begin{math}
	            AN[u] = AN(u)  \cap   \{u\}	     
	       \end{math} 
	     and $\lvert AN(u)\rvert $ represents is the cardinality of $AN(u)$.
The weight of a $k$RDF is \begin{math} \sum\limits_{{u \in V}}f(u) \end{math}. The \textit{$[k]$-Roman domination number} \begin{math} \gamma_{[kR]} (G)\end{math} is the minimum weight of all the $k$RDFs on $G$. \\ 
Motivated by the Integer Linear Programming (ILP) formulations given for different domination problems \cite{ilp}  \cite{milp} \cite{ilprdp} \cite{ilpwrdp}, in this paper we present different  ILP formulations for M3RDP and M4RDP problems. The proposed models are implemented using IBM CPLEX 22.1 optimization solvers and obtained results for random graphs generated using NetworkX Erdős-Rényi model, a widely used  model for generating random graphs  \cite{renyi1}  \cite{renyi}.

\section{ILP Formulations for  Triple Roman Domination Problem}\label{sec2}
   Proposed ILP formulations for M3RDP are discussed in this section.
 \subsection{M3RDP-1}\label{sec3}
        This is very basic ILP formulation for M3RDP where we assigned four decision variables (\begin{math}p_v,q_v,r_v,s_v\end{math}) for each vertex $v$ and their definitions are as follows
           \begin{center}
        \begin{math}
            p_v=\left\{
                \begin{array}{ll}
                  1, & \mbox{$f(v)$ = 1}.\\
                  0, & \mbox{otherwise}.
                \end{array}
            \right.
        \end{math} 
        \begin{math}
          q_v=\left\{
            \begin{array}{ll}
              1, & \mbox{$f(v)$ = 2}.\\
              0, & \mbox{otherwise}.
            \end{array}
          \right.
        \end{math} 
        \begin{math}
          r_v=\left\{
            \begin{array}{ll}
              1, & \mbox{$f(v)$ = 3}.\\
              0, & \mbox{otherwise}.
            \end{array}
          \right.
        \end{math} 
        \begin{math}
          s_v=\left\{
            \begin{array}{ll}
              1, & \mbox{$f(v)$ = 4}.\\
              0, & \mbox{otherwise}.
            \end{array}
          \right.
        \end{math} 
        \begin{math}
            t_v=\left\{
               \begin{array}{ll}
                 1, & \mbox{\begin{math}
                  \sum\limits_{{u\in N(v)}}q_u \ge 1
                 \end{math}}\\
                 0, & \mbox{otherwise}.
               \end{array}
            \right.
        \end{math}
        \begin{math}
            x_v=\left\{
               \begin{array}{ll}
                 1, & \mbox{\begin{math}
                  \sum\limits_{{u\in N(v)}}r_u \ge 1
                 \end{math}}\\
                 0, & \mbox{otherwise}.
               \end{array}
            \right.
        \end{math}
    \end{center}
An ILP formulation for M3RDP is given below:  
\begin{subequations} 
    \begin{equation}
    \hspace{-5mm} Minimize \hspace{10mm} \sum\limits_{{v\in V}} p_v + \sum\limits_{{v\in V}} 2q_v + \sum\limits_{{v\in V}} 3r_v + \sum\limits_{{v\in V}} 4s_v 
    \end{equation}
    Such that
    \begin{equation}
    p_v + q_v + r_v + s_v + \sum_{u \in N(v)}s_u + \frac{1}{3} \sum_{u \in N(v)}q_u + \frac{1}{2}\left( t_v + x_v \right) \ge 1
    \end{equation}

    \begin{equation}
     \frac{1}{2} \sum_{u \in N(v)}q_u + \sum_{u \in N(v)}\left ( s_u + r_u \right) \ge p_v 
    \end{equation}
    \begin{equation}
        \sum_{u \in N(v)}\left(q_u + r_u + s_u \right)
      \ge q_v
    \end{equation}
    \begin{equation}
       p_v + q_v + r_v + s_v
      \le 1
    \end{equation}
    \begin{equation}
       p_v , q_v , r_v , s_v ,t_v,x_v
      \in \{0,1\}
    \end{equation}
\end{subequations}
\\
As we can see number of decision variables are 4$\lvert V \rvert$ (\begin{math}p_v,q_v,r_v,s_v\end{math}) and number of constraints are 5$\lvert V \rvert $.

Constraint \textcolor{blue}{(1a)} gives the objective function value which is also the value of triple Roman domination number \begin{math}\gamma_{3R}(G)\end{math}. Constraint \textcolor{blue}{(1b)} is to ensure that every vertex with $f(v)=0$ should have at least 1 neighbour with $f(v)=4$ or at least 3 neighbours with
$f(v)=2$ or at least 2 vertices $v$,$w$ such that $f(v)$=2
and $f(w)=3$. Constraint \textcolor{blue}{(1c)} is to 
ensure that for every vertex $v$ with $f(v)=1$ there 
should be at least 2 neighbours with $f(v)=2$ or at least 1 neighbour with $f(v)\ge 3$.
Constraint \textcolor{blue}{(1d)} ensures that for every
vertex $v$ with $f(v)=2$ there should be at least 1 
neighbour with $f(v)\ge2$.
Constraint \textcolor{blue}{(1e)} is to ensure that every vertex gets only one label from the set \{0,1,2,3,4\}. Constraint
\textcolor{blue}{(1f)} ensures that variables \begin{math}p_v,q_v,r_v,s_v\end{math} does not take any other integer values other than 0 or 1, restricting them to be
decision variables.

\subsection{M3RDP-2}\label{sec4}
       In this section, we propose an ILP formulation for M3RDP with only three binary variables. Substituting $p_v$ = 0 in M3RDP-1 gives us a new model M3RDP-2.
      After reducing one variable the formulation will become
 \begin{subequations}
    \begin{equation}
        \hspace{-20mm} Minimize \hspace{10mm}
        \sum\limits_{{v\in V}} 2q_v + \sum\limits_{{v\in V}} 3r_v + \sum\limits_{{v\in V}} 4s_v 
    \end{equation}
    \hspace*{1cm}Such that
    
    \begin{equation} 
         q_v + r_v + s_v + \sum_{u \in N(v)}s_u + \frac{1}{3} \sum_{u \in N(v)}q_u +  \frac{1}{2}\left ( t_v + x_v \right)
     \ge 1 
    \end{equation}
    \begin{equation}
        \sum_{u \in N(v)}\left( q_u + r_u + s_u \right)
        \ge q_v
    \end{equation}
    \begin{equation}
        q_v + r_v + s_v
        \le 1
    \end{equation}
    \begin{equation}
        q_v , r_v , s_v ,t_v,x_v
        \in \{0,1\}
    \end{equation}
 \end{subequations}
 \\
Here the total number of decision variables are 3$\lvert V \rvert$ (\begin{math}q_v,r_v,s_v\end{math}) and the number of constraints are 4$\lvert V \rvert $. Let's say \begin{math} \gamma_{3R}^{'}(G) \end{math} is the weight of triple Roman domination function we get from M3RDP-2. \medskip \\
Next we show that ILP formulations M3RDP-1 and M3RDP-2 are equivalent except for the number of variables and constraints.
\begin{lemma}
 Let $h$ be a $\gamma_{3R}$-function  of $G$. Then there doesn't exist two adjacent vertices $v$,$u$ with labels as $h(v)=1$ and $h(u)=4$.
 \begin{proof}
By contradiction, let $s$ and $t$ be two adjacent vertices of $G$ with labels 1 and 4 under $h$. Now, we obtain a labelling $h'$ of $G$ from $h$ as follows.
\begin{equation} \label{3RDFwo14}
h'(v) =
\begin{cases}
	0, & \text{ if } \; v=s  \\
	4, & \text{ if } \; v=t \\
 h(v), & \text{otherwise}
\end{cases}          
\end{equation}  
\end{proof}
\end{lemma}

Clearly $h'$ is a 3RDF of $G$, because no other label depends on vertex with label 1. Hence $w(h') < w(h)$, a contradiction.
\begin{theorem} 
M3RDP-1 is equivalent to M3RDP-2.
\end{theorem}
\vspace{-5mm}
\begin{proof} 
Clearly $\gamma_{3R}(G)$ is the optimal value of M3RDP-1 for a given  graph $G$. Let $\gamma_{3R}'(G)$ is the optimal value of M3RDP-2. Next, we show that $\gamma_{3R}(G)=\gamma_{3R}^{'}(G)$. In M3RDP-2 we have removed the variable \begin{math} p_v \end{math} which represents the vertices with $f(v) = 1$. In fact we don't need \begin{math} p_v \end{math} that is we can always rearrange labels of vertices in the closed neighbourhood such the weight of the closed neighbourhood does not change. Let's see how we can rearrange. \\
 
Let's say vertex $v$ has $f(v)=1$ then the conditions are such that it should have at least 2 adjacent vertices whose $f(u)=2$ or at least 1 adjacent vertex whose \begin{math}f(u) \ge 3 \end{math}. Now lets consider the possible scenarios of neighbourhood of vertex $v$.

1) Vertex $v$ has two such adjacent vertices $u,w$ where $f(u)=2$ and $f(w)=2$. Now we rearrange labels of $v,u,w$ as below
\begin{center}
$f(v)=0,f(u)=3,f(w)=2$
\end{center}

we can create a 3RDF $g$ defined as follows: $g(x)=f(x)$ for all \begin{math} x \notin \{v,u\} \end{math} and $g(v)=0$, $g(u)=3$. $g$ is a 3RDF with same weight as $f$.

2) Our second scenario is that vertex $v$ has at least 1 adjacent vertex $u$ such that $f(u)=3$ then we can rearrange labels of $v$ and $u$ as below
\begin{center}
$f(v)=0,f(u)=4$
\end{center}

Modified $f$ is also a 3RDF.

3) As discussed in lemma 1, there will be no optimal labelling of adjacent vertices $v$,$u$ with $f(v)=1$ and $f(u)=4$.
\end{proof}

\subsection{M3RDP-3}\label{sec5}
Further optimising the above ILP formulation M3RDP-2 we can remove constraint (2d). So our new optimised formulation will be 
\begin{subequations} 
    
    \begin{equation}
        \hspace{-20mm} Minimize \hspace{10mm} \sum\limits_{{v\in V}} 2q_v + \sum\limits_{{v\in V}} 3r_v + \sum\limits_{{v\in V}} 4s_v
    \end{equation}
    \hspace*{1cm} Such that
    \begin{equation} 
        q_v + r_v + s_v + \sum_{u \in N(v)}s_u + \frac{1}{3} \sum_{u \in N(v)}q_u + \frac{1}{2}\left( t_v + x_v \right)
     \ge 1
    \end{equation}
    \begin{equation}
        \sum_{u \in N(v)} \left( q_u + r_u + s_u \right )
        \le q_v
    \end{equation}
    \begin{equation}
        q_v , r_v , s_v , t_v , x_v
        \in \{0,1\}
    \end{equation}
\end{subequations}
Here the number of decision variables are 3$\lvert V \rvert$ (\begin{math}q_v,r_v,s_v\end{math}) and number of constraints are 3$\lvert V \rvert$.

\begin{theorem}
 Let \begin{math} \gamma_{3R}^{'}(G) \end{math} and \begin{math} \gamma_{3R}^{''}(G) \end{math} represent the optimal weight of 3RDF obtained using M3RDP-2 and M3RDP-3. Then  \begin{math} \gamma_{3R}^{'}(G) = \gamma_{3R}^{''}(G) \end{math}
\end{theorem}
\vspace{-2mm}
\begin{proof}
    Since the absence of constraint (2d) in M3RDP-3 allows 
    $p_v,q_v,r_v,s_v$ to take 1 simultaneously i.e for example both 
     $r_v,s_v$ can be 1 which has no significance while calculating triple Roman domination number using those decision variables for vertex $v$. If $S_2$ is the set of all possible 3RDFs using M3RDP-2 formulation and  similarly $S_3$ is the set of all possible 3RDFs for M3RDP-3 then it is clear that 
     \begin{math}
         S_2 \subseteq S_3 
     \end{math}
     which is shown in figure 1.\\
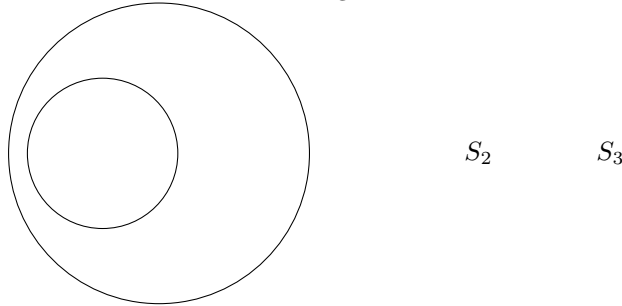
\begin{tikzpicture}
\hspace{5cm}
\draw (0,0) circle [radius=2] node {\hspace{2cm}$S_3$};
\draw(-0.75,0) circle [radius=1] node {$S_2$};
\end{tikzpicture}

From the above Venn diagram it is clear that any 3RDF given by M3RDP-3 is either from region \begin{math}S_2\end{math} or from region \begin{math}S_{32}(=S_3-S_2\end{math}). If we choose from $S_2$ there is no chance that we get optimal value lesser than M3RDP-2 formulation so we need to check from the set $S_{32}$. In $S_{32}$ at least 2 variables from \begin{math}q_v,r_v,s_v\end{math} will be equal to 1. Lets assume \begin{math}H(v)=2q_v+3r_v+4s_v\end{math}. For a vertex \begin{math}v \in V\end{math} and  \begin{math} f \in S_{32}\end{math}, the maximum value of  \begin{math}H(v)\end{math} will be 2+3+4=9 and minimum value of \begin{math}H(v)\end{math}  will be 5 with M3RDP-3 formulation where as if \begin{math} f \in S_2 \end{math} then maximum and minimum values of \begin{math}H(v)\end{math} are 4 and 2 respectively. We can see that maximum value \begin{math}H(v)\end{math} when \begin{math} f \in S_2\end{math} is less than that of minimum value \begin{math}H(v)\end{math} of when \begin{math} f \in S_{32}\end{math}. So definitely the optimal value with function from \begin{math}S_{32}\end{math} set will be strictly greater than the optimal value with function from \begin{math}S_2\end{math} set with M3RDP-3 formulation so we will be choosing function from \begin{math}S_2\end{math} set without the constraint (2d) also. From Theorem 1,  we know that M3RDP-2 gives optimal answer. Therefore from the fact that   M3RDP-3 chooses a function from set $S_2$ it follows that both the models are give the same result.
\end{proof}

\section{ILP Formulations for Quadruple Roman Domination Problem}\label{sec6}
In the similar way of construction of ILP models for M3RDP we have also provided ILP models for M4RDP.
\subsection{M4RDP-1}\label{sec7}
  \begin{center}
       \begin{math}
          p_v=\left\{
            \begin{array}{ll}
              1, & \mbox{$f(v)$ = 1}.\\
              0, & \mbox{otherwise}.
            \end{array}
          \right.
        \end{math} 
        \begin{math}
          q_v=\left\{
            \begin{array}{ll}
              1, & \mbox{$f(v)$ = 2}.\\
              0, & \mbox{otherwise}.
            \end{array}
          \right.
        \end{math} 
        \begin{math}
          r_v=\left\{
            \begin{array}{ll}
              1, & \mbox{$f(v)$ = 3}.\\
              0, & \mbox{otherwise}.
            \end{array}
          \right.
        \end{math} 
        \begin{math}
          s_v=\left\{
            \begin{array}{ll}
              1, & \mbox{$f(v)$ = 4}.\\
              0, & \mbox{otherwise}.
            \end{array}
          \right.
        \end{math}
\hspace{-0.5mm}
  \begin{math}
              t_v=\left\{
                \begin{array}{ll}
                  1, & \mbox{$f(v)$ = 5}.\\
                  0, & \mbox{otherwise}.
                \end{array}
              \right.
            \end{math} 
             \begin{math}
              x_v=\left\{
                \begin{array}{ll}
                  1, & \mbox{\begin{math}\sum\limits_{{u\in N(v)}} q_u \ge 1 \end{math}}.\\
                  0, & \mbox{otherwise}.
                \end{array}
              \right.
            \end{math} 
            \begin{math}
              y_v=\left\{
                \begin{array}{ll}
                  1, & \mbox{\begin{math}\sum\limits_{{u\in N(v)}} r_u \ge 1 \end{math}}.\\
                  0, & \mbox{otherwise}.
                \end{array}
              \right.
            \end{math} 
            \vspace{0.5cm}
            \begin{math}
              z_v=\left\{
                \begin{array}{ll}
                  1, & \mbox{\begin{math}\sum\limits_{{u\in N(v)}} s_u \ge 1 \end{math}}.\\
                  0, & \mbox{otherwise}.
                \end{array}
              \right.
            \end{math} 
             \begin{math}
              a_v=\left\{
                \begin{array}{ll}
                  1, & \mbox{\begin{math}\sum\limits_{{u\in N(v)}} q_u \ge 2 \end{math}}.\\
                  0, & \mbox{otherwise}.
                \end{array}
              \right.
            \end{math} 
        \end{center} 
        
\begin{subequations} 
    \begin{equation}
        \hspace{-20mm} 
        Minimize  \hspace{10mm} \sum\limits_{{v\in V}} p_v + 2\sum\limits_{{v\in V}}q_v+ 
        3\sum\limits_{{v\in V}}r_v + 4\sum\limits_{{v\in V}}s_v+ 5\sum\limits_{{v\in V}}t_v \\
    \end{equation}
    Such that
    \begin{equation} 
    \begin{split}
        p_v+q_v+r_v+s_v+t_v+\sum\limits_{{u \in N(V)}}t_u+\frac{1}{2}(x_v+z_v)+\frac{1}{2}(z_v+y_v)+\frac{1}{2}(a_v+y_v) \\+\frac{1}{2}\sum\limits_{{u \in N(v)}} s_u+\frac{1}{2} \sum\limits_{{u \in N(v)}} r_u
       +\frac{1}{4} \sum\limits_{{u \in N(v)}} q_u \ge 1
    \end{split}
    \end{equation}
    \begin{equation} 
      \sum\limits_{{u \in N(v)}}t_u+\sum\limits_{{u \in N(v)}}s_u+ \frac{1}{2} \left(y_v+x_v\right)+\frac{1}{2} \sum\limits_{{u \in N(v)}} r_u+\frac{1}{3} \sum\limits_{{u \in N(v)}} q_u \ge p_v
    \end{equation}
 \begin{equation} 
      \sum\limits_{{u \in N(v)}} \left (t_u+s_u+r_u \right)+\frac{1}{2}\sum\limits_{{u \in N(v)}}q_u \ge q_v
    \end{equation}
    \begin{equation} 
      \sum\limits_{{u \in N(v)}} \left ( t_u + s_u + q_u + r_u  \right ) \ge r_v
    \end{equation}
    \begin{equation}
        p_v + q_v + r_v + s_v + t_v
        \le 1
    \end{equation}
    \begin{equation}
       p_v, q_v , r_v , s_v , t_v , x_v
        \in \{0,1\}
    \end{equation}
\end{subequations}
As we can see number of decision variables are 5$\lvert V \rvert$ (\begin{math}p_v,q_v,r_v,s_v,t_v\end{math}) and number of constraints are 6$\lvert V \rvert$.

Constraint \textcolor{blue}{(4a)} gives the objective function value which is the quadruple Roman domination number \begin{math}\gamma_{4R}(G)\end{math}. Constraint \textcolor{blue}{(4b)} is to ensure that every vertex with $f(v)=0$ should have at least 1 neighbour with $f(v)=5$ or at least 2 vertices $v$,$w$ such that $f(v)=4$ and $f(w)\ge2$ or at least 2 vertices $v$,$w$ with $f(v)=3$ and $f(w)\ge3$ or at least four vertices with $f(v)=2$. Constraint \textcolor{blue}{(4c)} is to ensure that for every vertex $v$ with $f(v)=1$ there should be at least 1 vertex with $f(v)\ge4$ or at least 2 vertices $v$,$w$ with $f(v)=3$ and $f(w)\ge2$. Constraint \textcolor{blue}{(4d)} ensures that for every vertex $v$ with $f(v)=2$ there should be at least 1 
neighbour with $f(v)\ge3$ or at least 2 vertices with $f(v)\ge2$.
Constraint \textcolor{blue}{(4e)} is to ensure that for every vertex $v$ with $f(v)=3$ there should be at least 1 neighbour with $f(v)\ge2$. Constraint \textcolor{blue}{(4f)} is to ensure that every vertex gets only one label from the set \{0,1,2,3,4,5\}. Constraint \textcolor{blue}{(4g)} ensures that variables \begin{math}p_v,q_v,r_v,s_v,t_v,x_v\end{math} does not take any other integer values other than 0 or 1, restricting them to be decision variables.

\subsection{M4RDP-2}\label{sec8}
In this section, we proposed an ILP formulation of M4RDP with only 4 decision variables. Substituting $p_v$ = 0 in M4RDP-1 gives us a new model M4RDP-2. After reducing one variable the formulation will become
\begin{subequations} 
    \begin{equation}
        \hspace{-20mm} 
        Minimize \hspace{5mm} 2\sum\limits_{{v\in V}}q_v+ 3 \sum\limits_{{v\in V}}r_v + 4\sum\limits_{{v\in V}}s_v+ 5 \sum\limits_{{v\in V}}t_v \\
    \end{equation}
    Such that
    \begin{equation} 
    \begin{split}
         q_v+r_v+s_v+t_v+\sum\limits_{{u \in N(v)}}t_u+\frac{1}{2}(x_v+z_v)+\frac{1}{2}(z_v+y_v)+\frac{1}{2}(a_v+y_v)\\+\frac{1}{2}\sum\limits_{{u \in N(v)}} s_u+\frac{1}{2} \sum\limits_{{u \in N(v)}} r_u
       +\frac{1}{4} \sum\limits_{{u \in N(v)}} q_u \ge 1
    \end{split}
    \end{equation}
   \begin{equation} 
      \sum\limits_{{u \in N(v)}} \left (t_u+s_u+r_u \right)+\frac{1}{2}\sum\limits_{{u \in N(v)}}q_u \ge q_v
    \end{equation}
 \begin{equation} 
      \sum\limits_{{u \in N(v)}} \left ( t_u + s_u + q_u + r_u  \right ) \ge r_v
    \end{equation}
    \begin{equation}
        q_v + r_v + s_v + t_v
        \le 1
    \end{equation}
    \begin{equation}
       q_v , r_v , s_v , t_v , x_v
        \in \{0,1\}
    \end{equation}
\end{subequations}
\begin{lemma}
For optimal labelling of vertices for quadruple Roman domination function, there doesn't exist a labelling where two adjacent vertices $v$,$u$ labels as $f(v)=1$ and $f(u)=5$.
\end{lemma}
\begin{proof}
Similar to lemma 1, there will be no such optimal combination since we can change vertex $v$ label to 0 because no other label depends on vertex with label 1.Changing label of vertex $v$ from 1 to 0 decreases weight of the function $f$ and resulting in more optimal value. 
\end{proof}

\begin{theorem} 
M4RDP-1 is equivalent to M4RDP-2.
\end{theorem}
\vspace{-5mm}
\begin{proof} 
Similar to theorem 1, in M4RDP-2 we have removed the variable \begin{math} p_v \end{math} which represents the vertices with $f(v) = 1$. Rearranging labels can be done as follows.

Let's say vertex $v$ has $f(v)=1$ then the conditions are such that it should have at least 1 adjacent vertex $u$ whose $f(u)=5$ or at least 1 adjacent vertex $u$ whose $f(u)=4$ or at least one adjacent vertex $u$ with $f(u)=3$ and at least adjacent vertex $w$ with $f(w)=2$ or at least 3 adjacent vertices with label.
\begin{math}f(u) \ge 3 \end{math}. Now lets consider the possible scenarios of neighbourhood of vertex $v$.\\
1) From lemma 2, there will no optimal labelling of adjacent vertices $v,u$ with $f(v)=1$ and $f(u)=5$.\\
2) Vertex $v$ has at least adjacent vertex $u$ where $f(u)=4$ . Now we rearrange labels of $v,u,w$ as below
\begin{center}
$f(v)=0,f(u)=5$
\end{center}

we can create a triple Roman domination function $g$ defined as follows: $g(x)=f(x)$ for all \begin{math} x \notin \{v,u\} \end{math} and $g(v)=0$, $g(u)=5$. $g$ is a 4RDF with same weight as $f$.\\
3) Vertex $v$ has at least 1 adjacent vertex $u$ such that $f(u)=3$ and at least 1 adjacent vertex $w$ such that $f(w)=2$ then we can rearrange labels of $v$,$u$ and $w$ as below
\begin{center}
$f(v)=0,f(u)=3,f(w)=3$
\end{center}

Modified $f$ also 4RDF of $G$.\\
4) Vertex $v$ has at least 3 adjacent vertex $u,w,y$ such that $f(u)=2$ then we can rearrange the labels of $v$,$u$,$w$ and $y$ as below
\begin{center}
$f(v)=0,f(u)=2,f(w)=2,f(y)=3$
\end{center}
 \end{proof}
Here number of decision variables are 4$\lvert V \rvert$ (\begin{math}q_v,r_v,s_v,t_v\end{math}) and number of constraints are 5$\lvert V \rvert$.
\subsection{M4RDP-3}\label{sec9}
Further optimising the above ILP formulation M4RDP-2, we can remove constraint (5e). So our new optimised formulation will be 
    
\begin{subequations} 
    \begin{equation}
        \hspace{-20mm} 
        Minimize \hspace{5mm} 2\sum\limits_{{v\in V}}q_v+ 3 \sum\limits_{{v\in V}}r_v + 4\sum\limits_{{v\in V}}s_v+ 5 \sum\limits_{{v\in V}}t_v \\
    \end{equation}
    Such that
    \begin{equation}
    \begin{split}
      q_v+r_v+s_v+t_v+\sum\limits_{{u \in N(v)}}t_u+\frac{1}{2}(x_v+z_v)+\frac{1}{2}(z_v+y_v)+\frac{1}{2}(a_v+y_v) \\ +\frac{1}{2}\sum\limits_{{u \in N(v)}} s_u 
       +\frac{1}{2} \sum\limits_{{u \in N(v)}} r_u
       +\frac{1}{4} \sum\limits_{{u \in N(v)}} q_u \ge 1
       \end{split}
    \end{equation}
 \begin{equation} 
      \sum\limits_{{u \in N(v)}} \left (t_u+s_u+r_u \right)+\frac{1}{2}\sum\limits_{{u \in N(v)}}q_u \ge q_v
    \end{equation}
    \begin{equation} 
      \sum\limits_{{u \in N(v)}} \left ( t_u + s_u + q_u + r_u  \right ) \ge r_v
    \end{equation}
    \begin{equation}
       q_v , r_v , s_v , t_v , x_v
        \in \{0,1\}
    \end{equation}
\end{subequations}
Here number of decision variables are 4$|V|$ (\begin{math}q_v,r_v,s_v,t_v\end{math}) and number of constraints are 4 $|V|$. Let's say \begin{math} \gamma_{4R}^{''}(G) \end{math} is the optimal value of quadruple Roman domination number we get from M4RDP-3. Similar to theorem 2, we can also prove that \begin{math} \gamma_{4R}^{''}(G) \end{math} is same as \begin{math} \gamma_{4R}(G) \end{math} 
\section{Computational Results}\label{sec10}
In this section, we showed the results of all the models we introduced in section 3. All the models are solved using IBM CPLEX 22.1 optimization solvers. NetworkX, a python package is used to generate random graphs. Whole computaion is performed on MacOS Monterey version 12.6, M1 chip 2020, 8GB RAM.
Random graphs are generated using NetworkX Erdős–Rényi model. We have 2 tables representing the running time of our models on the random graphs. Table 1 contains 4 columns instance refers to properties of graph $\lvert V \rvert $ means number of vertices, $\lvert E \rvert$ number of edges, $p$ is the probability of each edge whether to take into generation of graph or not while generating random graphs. Second column represent our first model in section 2 and it contains 2 sub columns result and time, result is \begin{math}\gamma_{3R}(G)\end{math} and time is the time taken to get solution using M3RDP-1. Similarly next columns represent our remaining model      performances. We can see that all graphs are giving same results for same graph which represents correctness of our models. Similarly Table-2 represents our ILP model performances for minimum quadruple Roman domination problem. We are using '-' sign to represent that optimal solution is not being found in 1800 seconds for that particular model. We can see that for most of the cases M3RDP-3 model performs better than the other two models for triple Roman domination problem. For quadruple Roman domination problem, M4RDP-3 performs better than remaining two.

\begin{table*}[h] 
        \centering
        \caption{Performances of ILP Model for triple Roman domination problem}       
            
        \begin{tabular}{l c c c c c c c c c}
        \toprule
        \multicolumn{3}{c}{\small \textbf{Instance}}&
        \multicolumn{2}{c}{\small\textbf{M3RDP-1}} &
        \multicolumn{2}{c}{\small\textbf{M3RDP-2}} &
        \multicolumn{2}{c}{\small \textbf{M3RDP-3}} \\ 
        \hline
        $ \mid V \mid $&$ \mid E \mid $&p&result&time&result&time&result&time\\
        \hline
       25 & 50 &
  0.2 & 26 & 0.44 & 26 & 0.38 &  26 & 0.39 \\
  &
  141 &
  0.5 &
  11 &
  0.45 &
  11 &
  0.43 &
  11 &
  0.81 \\
 &
  248 &
  0.8 &
  7 &
  0.37 &
  7 &
  0.43 &
  7 &
  0.35 \\
\hline
  50 & 250 &
  0.2 &
  24 &
  14.05 &
  24 &
  5.82 &
  24 &
  3.95 \\
 &
  631 &
  0.5 &
  12 &
  174.49 &
  12 &
  330.18 &
  12 &
  89.37 \\
 &
  985 &
  0.8 &
  8 &
  3.51 &
  8 &
  2.04 &
  8 &
  2.42 \\
  \hline
    75 & 556 &  0.2 & - & - & -& -& - & -  \\ &
    1275 & 0.5	& - &	- & 12 & 51.85 & 12 & 156.68 & \\ &
    2238 & 0.8 & 8 & 23.79 & 8 & 6.16 & 8 & 11.75 \\
 \hline
    
    100 & 1015 & 0.2	& - & - & -& -& - & - \\ &
    2502 & 0.5 & - & - & -& -& - & -  \\ & 
    3961 & 0.8 & 8 & 17.02 & 8 & 16.72 & 8 & 19.36 \\
\hline
    150 & 2257 & 0.2	& - & - & -& -& - & - \\ &
    5530 & 0.5 & - & - & -& -& - & -  \\ & 
    8973 & 0.8 & 8 & 16.52 & 8 & 14.22 & 8 & 15.45 \\
\hline
        \end{tabular} 
\end{table*}

\begin{table*}
         \centering
        \caption{Performances of ILP Model for quadruple Roman domination problem}       
            
        \begin{tabular}{l c c c c c c c c c}
        \toprule
        \multicolumn{3}{c}{\small \textbf{Instance}}&
        \multicolumn{2}{c}{\small\textbf{M4RDP-1}} &
        \multicolumn{2}{c}{\small\textbf{M4RDP-2}} &
        \multicolumn{2}{c}{\small \textbf{M4RDP-3}} \\ 
        \hline
         $ \mid V \mid $ & $\mid E \mid $ &p&result&time&result&time&result&time\\
        \hline
        10 & 10  & 0.2 & 18 & 1.15 & 18 & 1.42   & 18 & 0.91   \\
   & 26  & 0.5 & 11 & 4.38 & 11 & 1.29   & 11 & 0.94   \\
   & 31  & 0.8 & 9  & 3.09 & 9  & 0.89   & 9  & 0.81   \\
   \hline
25 & 56  & 0.2 & -  & -   & 25 & 760.52 & 25 & 642.75 \\
   & 145 & 0.5 & -  & -    & 14 & 1247.1 & -  & -     \\
   & 244 & 0.8 & -  & -    & 9  & 76.41  & 9  & 61.31  \\
  \hline
50 & 242 & 0.2 & -  & -   & -  & -     & -  & -    \\
   & 601 & 0.5 & -  & -   & -  & -    & - & -   \\
   & 999 & 0.8 & - & -   & 9  & 1416   & - & -    \\
  \hline
        \end{tabular} 
\end{table*}
\pagebreak
\section{Conclusion}\label{sec11}
It is evident from both the tables that M3RDP-3 model is outperforming in most of the cases for M3RDP and M4RDP-3 model is outperforming for M4RDP.
As can be seen from Table-1, models are becoming infeasible from 100 vertices for M3RDP where as in Table-2 models are becoming infeasible from 50 vertices. From both the tables we can conclude that if the number of vertices are increasing, models are becoming infeasible within time limit.
Hence, providing ILP formulations for M3RDP and M4RDP with lesser number of constraints so that CPLEX solvers can successfully find solutions of larger
graph instances is an interesting direction for future work.\\
\bibliography{sn-bibliography}

\end{document}